\documentclass[11pt,a4paper]{article}
\usepackage[latin1]{inputenc}
\usepackage{amsmath}
\usepackage{amsfonts}
\usepackage{amssymb}
\usepackage{amsthm}
\usepackage{graphicx}
\usepackage{enumitem}
\author{Tuomas Hakoniemi\\
	Universitat Polit\`{e}cnica de Catalunya}
\title{Monomial-size vs. Bit-complexity in Sums-of-Squares and Polynomial Calculus}

\newcommand{\Q}{\mathbb{Q}}
\newcommand{\R}{\mathbb{R}}
\newcommand{\KNAPSACK}{\mathrm{KNAPSACK}}
\newcommand{\ks}{\mathrm{ks}}
\newcommand{\poly}{\mathrm{poly}}

\newtheorem{theorem}{Theorem}
\newtheorem{lemma}{Lemma}
\newtheorem{corollary}{Corollary}

\theoremstyle{definition}

\newtheorem{definition}{Definition}
\begin{document}
	\maketitle
	
	\begin{abstract}
		In this paper we consider the relationship between monomial-size and bit-complexity in Sums-of-Squares (SOS) in Polynomial Calculus Resolution over rationals (PCR/$\Q$). We show that there is a set of polynomial constraints $Q_n$  over Boolean variables that has both SOS and PCR/$\Q$   refutations of degree 2 and thus with only polynomially many monomials, but for which any SOS or PCR/$\Q$   refutation must have exponential bit-complexity, when the rational coefficients are represented with their reduced fractions written in binary.
	\end{abstract}
	
	\section{Introduction}
	
	%This paper studies the algebraic proof system Polynomial Calculus Resolution over the field of rationals (PCR/$\Q$) and the semialgebraic proof system Sums-of-Squares (SOS). 
	Polynomial Calculus (PC) is an algebraic proof system, introduced in \cite{CleggEI96}, that draws inspiration from Gr\"{o}bner basis computations in computational algebraic geometry. Ultimately it is based on Hilbert's Nullstellensatz, and Polynomial Calculus can be seen as a dynamic version of Nullstellensatz that is based on schematic inference rules. In the Boolean realm the original definition of \cite{CleggEI96} was strengthened to a system called Polynomial Calculus Resolution (PCR) in \cite{AlekhnovichBRW02} by the introduction of the so called twin variables, that allow for more compact representations of Boolean functions.
	
	Sums-of-Squares (SOS) on the other hand is based on Putinar's Positivstellensatz \cite{Putinar} in real algebraic geometry. As a refutation system it was originally considered by Grigoriev and Vorobjov in \cite{GrigorievV01}. On the other side of things SOS has strong connections to approximation algorithms through the hierarchies of semidefinite programming relaxations in combinatorial optimization \cite{Lasserre,Parrilo,BarakBHKSZ12}. We refer the reader to the survey \cite{Laurent2009} for further discussion.
	
	The most studied complexity measure for these systems has been the degree of a refutation. We on the other hand concentrate here on other possible complexity measures for the systems that are of interest especially to proof complexity. The measures studied here are monomial-size, the number of monomials in a refutation, and bit-complexity, the number of bits actually needed to write down the refutation.
	
	We encounter here a somewhat surprising situation where there is an exponential separation between monomial-size and bit-complexity. We obtain strong exponential bit-complexity lower bounds for a set of constraints that appears easy for the proof systems by the standards of the other complexity measures, degree and monomial-size. To the best of our knowledge this is the first lower bound of this sort for the systems in question.
	
	The question about bit-complexity of Sums-of-Squares proofs was raised by O'Donnell in \cite{ODonnell2017} in connection with (degree)-automatability of Sums-of-Squares (SOS) proofs. O'Donnell noted that the received wisdom, that a degree $d$   Sums-of-Squares proof can be found using the ellipsoid algorithm in time $n^{O(d)}$   if one exists, is not entirely true. Difficulties may arise if the only proofs of degree $d$   contain exceedingly large coefficients as then the initial ellipsoid cannot be chosen small enough to guarantee polynomial runtime.
	
	Building on O'Donnell's example, Raghavendra and Weitz exhibited in \cite{RaghavendraWeitz2017} (see also \cite{Weitz:EECS-2017-38}) an example of a set of polynomial constraints over $O(n^2)$   Boolean variables and a polynomial that has SOS proofs of degree $2$   from the constraints, but for which any SOS proof of degree $O(n)$   from the constraints must contain a coefficient of magnitude doubly exponential in $n$.
	
	The example of Raghavendra and Weitz leaves open, however, the possibility that there are SOS proofs of degree higher than $O(n)$, but with only polynomially many monomials, and with coefficients of polynomial bit-complexity. In other words the example leaves open the possibility that there are SOS proofs from the constraints that can be written down using only polynomially many bits.
	%We, in turn, build on \cite{RaghavendraWeitz2017}, but opt to prove our main result for an unsatisfiable set of constraints, since this leads to both cleaner statement and cleaner bounds in our main theorem. 
	
	Similarly for Polynomial Calculus it has been often stated that a degree $d$   Polynomial Calculus proof can be found in time $n^{O(d)}$ if they exist. This is certainly true for Polynomial Calculus over finite fields as is already clear from the proof search algorithm given in \cite{CleggEI96}. However for Polynomial Calculus over reals or rationals there is potential for significant coefficient bloat in the proof search.
	
	We give here an example of an unsatisfiable set of polynomial constraints that has both SOS and PCR/$\Q$   refutations of degree $2$, and thus with only polynomially many monomials, but for which any SOS or PCR/$\Q$   refutation must be of exponential bit-complexity. For SOS we prove this by showing that any refutation with less than exponentially many monomials must contain a coefficient of doubly exponential magnitude. For PCR/$\Q$   we show that any proof with bounded coefficients using less than exponentially many monomials must be exceedingly tall. 
	
	Raghavendra and Weitz proved the lower bound on the magnitude of coefficients in proofs of degree $O(n)$   using the linear degree lower bounds for refutations of Knapsack obtained in \cite{Grigoriev2001a}, and the linear degree pseudoexpectations for Knapsack it provides. We on the other hand will use lower bounds on the number of monomials in refutations of Knapsack, and a suitably generalized notion of pseudoexpectation that can be used to argue for bounds on number of monomials rather than degree.
	% 
	%To the best of our knowledge, this is the first lower bound of this sort for Sums-of-Squares, giving exponential separation between monomial size and bit-complexity of Sums-of-Squares proofs.
	
	\section{Preliminaries}
	
	We denote by $[n]$ the set of positive integers up to and including $n$, i.e. $[n] := \{1,\ldots,n\}$, and by $\R_+$ the set of non-negative real numbers. Given a polynomial $p$ in real or rational coefficients, we denote by $\Vert p \Vert$ the maximum coefficient of $p$ in absolute value.
	
	\subsection{Polynomials and the Boolean ideal}
	
	We consider real polynomials over $n$   pairs of Boolean variables $x_i,\bar{x}_i$, $i\in [n]$. The intended meaning here is that all variables take value $0$ or $1$ and the two variables in a pair assume the opposite values. Accordingly, we denote by $I_n$ \textbf{the Boolean ideal} over $n$   Boolean variables, i.e.
	$I_n := \langle x_i^2 - x_i,x_i + \bar{x}_i - 1: i\in [n]\rangle$. We write $p\equiv q \mod I_n$, when $p - q\in I_n$.
	
	Given a set $S$   of monomials in the variables $x_i,\bar{x}_i$   we denote by $\R[S]$   the set of all linear combinations of elements of $S$   with real coefficients, i.e. polynomials that use only monomials from the set $S$. We denote by $S^2$ the set of products $m_1m_2$, where $m_1,m_2\in S$.
	
	\subsection{Sums-of-Squares proofs}
	
	Let $Q = \{q_1 = 0,\ldots,q_m = 0 \}$   be a set of polynomial equality constraints. Let $p$ be another polynomial. A \textbf{Sums-of-Squares (SOS) proof} of non-negativity of $p$   from $Q$   is a polynomial equality of the form
	\begin{equation}\label{proof}
	p = \sum_{i\in [k]}r_i^2 + \sum_{q\in Q} t_q q 
	+ \sum_{i\in [n]}\left(u_i(x_i^2 - x_i) + v_i(x_i + \bar{x}_i - 1)\right),
	\end{equation}	
	where $r_i$, $t_q$, $u_i$   and $v_i$   are arbitrary polynomials. We call the polynomials $t_q$ the \textbf{lifts of the non-logical axioms}, and the polynomials $u_i$ and $v_i$ the \textbf{lifts of the logical axioms}.
	An \textbf{SOS refutation} of $Q$   is a proof of non-negativity of $-1$   from $Q$.
	%An SOS proof is a certificate that $p$ is non-negative on all the $\{0,1\}$-solutions of the equations $q = 0$ for all $q\in Q$.  
	%When situation permits, we do not write the explicit lifts of the Boolean axioms and simply write the proof \eqref{proof} as 
	%\[
	%p \equiv \sum_{i\in [k]}a_i r_i^2 + \sum_{q\in Q} t_q q \mod I_n.
	%\]
	
	%We include the scalars $a_i$   in the proof \eqref{proof} as these are important when we consider SOS proofs with rational coefficients. This addition ensures that we can always multiply a rational SOS proof with a rational number to obtain another rational SOS proof. For SOS proofs with real coefficients the scalars $a_i$ are of course redundant, since any non-negative real itself is a square of another real number, and can be thus pushed inside the square $r_i^2$.
	
	The proof \eqref{proof} has \textbf{degree} at most $d$   if $\deg(r_i^2)\leq d$   for any $i\in [k]$, $\deg(t_q q)\leq d$   for any $q\in Q$   and 
	and $\deg(u_i(x_i^2 - x_i))\leq d$   and $\deg(v_i(x_i + \bar{x}_i - 1))\leq d$   for any $i\in [n]$. 
	
	The \textbf{explicit monomials} of the proof \eqref{proof} are all the monomials appearing in the polynomials $r_i$   for any $i\in [k]$, in the polynomials $t_q$   and $q$   for any $q\in Q$   and in the polynomials $u_i$, $v_i$, $x_i^2 - x_i$   and $x_i + \bar{x}_i - 1$ for any $i\in [n]$. In other words, the explicit monomials of the proof \eqref{proof} are all the monomials visible in the explicit representation of the proof. 
	
	We distinguish a particular subset of all the explicit monomials as particularly important, and call the monomials in the polynomials $r_i$ and $t_q$ the \textbf{significant monomials} of the proof \eqref{proof}.
	
	The \textbf{monomial-size} of the proof \eqref{proof} is the number of explicit monomials counted with multiplicity. In the case that the polynomials in the proof \eqref{proof} have only rational coefficients, we define the \textbf{bit-complexity} of \eqref{proof} as the minimum length of a bit-string representing the proof when the rational coefficients are represented with their reduced fractions written in binary.
	
	%The measures define above do not care about the lifts of the Boolean axioms. For our purposes this actually is for our benefit as we are interested in lower bounds. Any additional information would only increase the complexity measures we define.
	
	Given a set $S$   of monomials that includes the empty monomial $1$, and all the monomials appearing in $Q$, we say that the proof \eqref{proof} is a \textbf{proof over $S$}, when all the significant monomials of \eqref{proof} are among $S$.
	% Note that we do not restrict the lifts of the Boolean axioms in this definition. This is because we want that any $p,q\in\R[S^2]$   that are equivalent modulo $I_n$   actually have SOS proofs over $S$ of their equivalence, i.e. proofs of both non-negativity of $p- q$   and $q - p$. 
	
	%This will be important in the proof of Lemma \ref{key-lemma} below. 
	
	We write $Q\vdash_S p\geq q$   if there is a proof of non-negativity of $p-q$   from $Q$   over $S$.
	
	\subsection{Polynomial Calculus Resolution}
	
	Let again $Q = \{q_1 = 0,\ldots,q_m = 0\}$   be a set of polynomial equality constraints, and let $p$ be another polynomial. A \textbf{Polynomial Calculus Resolution proof} over rationals (PCR/$\Q$) of $p = 0$   from $Q$   is a sequence 
	\begin{equation}\label{PCR-proof}
	p_1,\ldots,p_\ell
	\end{equation}
	of polynomials such that $p_\ell = p$   and for any $i\in [\ell]$   one of the following holds:
	\begin{itemize}
		\item $p_i\in Q$, i.e. $p_i$   is one of the non-logical axioms;
		\item $p_i\in B_n$, i.e. $p_i$ is one of the logical axioms;
		\item there is $j < i$ and a variable $x$   such that $p_i = xp_j$, i.e. $p_i$   is obtained from $p_j$   via lifting with a variable $x$;
		\item there are $j,k < i$   and $a,b\in\Q$   such that $p_i = ap_j + bp_k$, i.e. $p_i$   is obtained from $p_j$   and $p_k$   via linear combination.
	\end{itemize}\
	A \textbf{PCR/$\Q$ refutation }of $Q$ is a proof of $1 = 0$ from $Q$.
	
	The proof \eqref{PCR-proof} is of \textbf{degree} at most $d$   if $\deg(p_i)\leq d$ for every $i\in [\ell]$, and of \textbf{monomial-size} at most $s$ if the number of monomials counted with multiplicity in all the polynomials $p_i$ is at most $s$. The \textbf{height} of the proof \eqref{PCR-proof} is the length of the longest path from a leaf to $p_\ell$, when the proof \eqref{PCR-proof} is considered as a directed acyclic graph. The height of a polynomial $p_i$   in the proof \eqref{PCR-proof} is similarly the length of the longest path from a leaf to $p_i$. Hence all elements of $Q$ and $B_n$ are at height $0$.
	
	We define the \textbf{bit-complexity} of the proof \eqref{PCR-proof} to be the minimum length of a bit-string representing all the polynomials in the proof and all the scalars used in the linear combination rule, when all the coefficients and scalars are represented with their reduced fractions written in binary.
	
	\subsection{Size-degree trade-offs for Sums-of-Squares}
	
	Finally we recall the size-degree trade-offs for SOS obtained in \cite{AtseriasHakoniemi2019}, and discuss a minor improvement that can be made for our purposes.
	
	We showed in \cite{AtseriasHakoniemi2019} how to transform an SOS refutation with few significant monomials\footnote{In the paper \cite{AtseriasHakoniemi2019} we called the significant monomials of this paper by the name \emph{explicit monomials}.} into an SOS refutation of relatively low degree. In detail we showed that if a set $Q$   of polynomial constraints in $n$   pairs of twin variables has an SOS refutation with $s$   many significant monomials counted with multiplicity, then it has an SOS refutation of degree at most $4\sqrt{2(n+1)\log(s)}+k+4$, where $k$   is the maximum degree of the constraints in $Q$. This implies a degree-criterion for monomial-size lower bounds in SOS.
	
	The size-degree trade-off can however be slightly improved as we don't actually need to count the number of significant monomials with multiplicity for the proof to work. Thus we have the following.
	
	\begin{theorem}\label{size-degree}
		For every set $Q$   of equality constraints of degree at most $k$   in $n$   pairs of twin variables, if there is an SOS refutation of $Q$   with $s$   many \textbf{distinct} significant monomials, then there is an SOS refutation of $Q$   of degree at most $4\sqrt{2(n+1)\log(s)}+k+4$.
	\end{theorem}
	
	The proof of this theorem is truly almost identical to the proof of the main theorem of \cite{AtseriasHakoniemi2019}: one simply needs to add the word distinct to few places in the proof. We omit the details here.
	
	An important consequence of Theorem \ref{size-degree} is the following degree-criterion for the number of significant monomials in a refutation.
	
	\begin{corollary}\label{degree-criterion}
		For every set $Q$   of equality constraints of degree at most $k$ in $n$ pairs of twin variables, if $d_Q$   is the minimum degree of an SOS refutation of $Q$, $s_Q$   is the minimum number of distinct significant monomials needed to refute $Q$   in SOS, and $d_Q\geq k + 4$, then
		\[
		s_Q \geq \exp\left(\left(d_Q - k - 4\right)^2/\left(32\left(n + 1\right)\right)\right).
		\]
	\end{corollary}

	\section{$S$-pseudoexpectations}
	
	In this section we give a characterization for SOS refutations over a set of monomials using the so called $S$-pseudoexpectations. We will prove that $S$-pseudoexpectations give, in a sense, sound and complete semantics for SOS refutations over $S$.
	
	%These objects can be seen as giving semantical description for refutations with bounded resources. In this vein, we prove a soundness and completeness theorem for $S$-pseudoexpectations. 
	Pseudoexpectations have traditionally been used to argue against low-degree refutations and to prove degree lower bounds in SOS. They are linear functionals that fool degree bounded SOS to think that the set of constraints is satisfiable by mapping anything provably non-negative in bounded degree to a non-negative value. Probably the first instance of this idea appears in \cite{Grigoriev2001b}, however the term `pseudoexpectation' appears for the first time in \cite{BarakBHKSZ12}. Actually, the existence of degree bounded pseudoexpectations is equivalent to the non-existence of degree bounded refutations \cite{BarakSteurer2014}. Thus from a logical point of view pseudoexpectations can be considered as sound and complete semantics for degree bounded SOS refutations.
	
	$S$-pseudoexpectations generalize the notion of degree bounded pseudoexpectations to work against refutations over fixed sets of monomials. The idea here is simple: rather than defining the functionals on a vector space of polynomials up to some degree, we define the functionals on a vector space of polynomials over a given set of monomials. This allows us to use pseudoexpectations to argue for upper and lower bounds on monomial size of SOS refutations. Pseudoexpectations against refutations over a set of monomials have previously appeared in \cite{Hakoniemi2020}, however in a slightly different form. 
	
	%The results in this section are thus not new, but for the pleasure of the reader we keep this section self-contained.
	
	\begin{definition}
		A linear functional $E\colon \R[S^2]\rightarrow\R$   is an $S$-pseudoexpectation for $Q$   when
		\begin{itemize}
			\item $E(1) = 1$;
			\item $E(p)\geq 0$   if $Q\vdash_S p\geq 0$   for $p\in\R[S^2]$.
		\end{itemize}	
	\end{definition}
	
	For the soundness and completeness theorem we need the following lemma. 
	%We omit the proof as it is exactly the same as the proof of Lemma \ref{key-lemma-rationals}.
	%
	%The following lemma is a key in proving the soundness and completeness theorem. This lemma shows that for any $p\in\R[S^2]$ there is a real upper bound provable from no non-logical axioms.
	
	\begin{lemma}\label{key-lemma}
		For any $p\in\R[S^2]$   there is $r\in \R_+$   such that 
		\[
		\emptyset\vdash_S r \geq p.
		\]
	\end{lemma}
	
	\begin{proof}
		Let first $m\in S$, and let $a\in\R$. We want to show that there is some $b\in \R_+$   such that $\emptyset\vdash_S b\geq am$. If $a < 0$, then $-am \equiv {-a}m^2 \mod I_n$   and so $\emptyset\vdash_S 0 \geq am$. On the other hand if $a > 0$, then $a - am \equiv a(1 - m)^2 \mod I_n$, and so $\emptyset\vdash_S a\geq am$.
		
		Let then $m_1,m_2\in S$   and $a\in\R$. We want to show that there is some $b\in\R_+$   such that $\emptyset\vdash_S b\geq am_1m_2$. If $a < 0$, then $-am_1 - 2am_1m_2 - am_2 \equiv -a(m_1 + m_2)^2 \mod I_n$. On the other hand, by the above paragraph, there are $b_1,b_2\in\R_+$   such that $\emptyset\vdash_S b_i\geq -am_i$   for $i = 1,2$. Hence $\emptyset\vdash_S (b_1 + b_2)/2\geq am_1m_2$. If $a > 0$, then $am_1 - 2am_1m_2 + am_2 \equiv a(m_1 - m_2)^2 \mod I_n$. Again there are $b_1,b_2\in\R_+$   such that $\emptyset\vdash_S b_i\geq am_i$   for $i = 1,2$, and so $\emptyset\vdash_S (b_1 + b_2)/2\geq am_1m_2$.
	\end{proof}\
	
	Using the above lemma, we can prove the soundness and completeness theorem easily. Below, and later in Lemma \ref{simulation}, we need the hyperplane separation theorem in the following form: in a finite-dimensional real vector space $V$, given a convex cone $C$   and a convex set $D$   disjoint from $C$, there is a non-trivial, i.e. one that takes also non-zero values, linear functional $L\colon V\rightarrow\R$   such that $L(c)\geq 0$   for any $c\in C$   and $L(d)\leq 0$ for every $d\in D$. That is, there is always a hyperplane separating a convex cone and a convex set that passes through the origin. 
	
	\begin{theorem}[Soundness and Completeness]\label{soundnesscompleteness} 
		Let $S$   be a finite set of monomials. Then there is no SOS refutation of $Q$   over $S$   if and only if there is an $S$-pseudoexpectation for $Q$.
	\end{theorem}
	
	\begin{proof}
		It is clear that if there is a refutation of $Q$   over $S$, then there cannot exist an $S$-pseudoexpectation for $Q$, since any $S$-pseudoexpectation maps any proof over $S$   to a non-negative value, and $-1$  to $-1$.
		
		Suppose then that there is no refutation of $Q$  over $S$. Then $-1$ is not in the convex cone of all polynomials in $\R[S^2]$  provably non-negative from $Q$  over $S$, and so, by the hyperplane separation theorem, there is a non-trivial linear functional $L\colon\R[S^2]\rightarrow\R$  such that $L(p)\geq 0$  if $Q\vdash_S p\geq 0$  and $L(-1)\leq 0$. 
		
		We argue first that $L(1) > 0$. So suppose towards a contradiction that $L(1) = 0$. Now, by Lemma \ref{key-lemma} for any $p\in\R[S^2]$  there is non-negative $r$  such that $Q\vdash_S r\geq p\geq -r$. Hence $L(r)\geq L(p)\geq L(-r)$. But now, by assumption, $L(r) = L(-r) = 0$, and so $L(p) = 0$, against the non-triviality of $L$.
		
		Now by defining $E(p) = L(p)/L(1)$, we have an $S$-pseudoexpectation for $Q$.
	\end{proof}\
	
	We remark that we can actually remove the assumption about the finiteness of $S$  in Theorem \ref{soundnesscompleteness}. However we only need the theorem in the form stated and thus omit the details of the more general statement.
	
	%	As a corollary we obtain the following duality theorem using Theorem 4 of \cite{AtseriasHakoniemi2019}.
	%	
	%	\begin{theorem}
	%		For any $p\in \R[S^2]$,
	%		\[
	%		\sup\{c\in\R : Q\vdash_S p \geq c\} = \inf\{E(p) : E\in\mathcal{E}_S(Q)\}.
	%		\]
	%		Moreover if $\mathcal{E}_S(Q)\neq \emptyset$, then the infimum is attained.
	%	\end{theorem}
	%	
	%	In particular we have the following soundness and completeness property for $S$-pseudoexpectations: there is an SOS refutation of $Q$ over $S$  if and only if $\mathcal{E}_S(Q) = \emptyset$. 
	
	\section{The constraints}\label{section:knapsack}
	
	In this section we introduce the set of constraints that we use to prove our claims about lower bounds on bit-complexity.
	
	Recall first the knapsack constraint  $\KNAPSACK(n,k)$: 
	\[x_1 + \ldots + x_n = k.\] 
	If $k$  is not an integer, the constraint is unsatisfiable over the Boolean values. However, for any $\varepsilon$  strictly between $0$  and $1$, $\KNAPSACK(2n,n + \varepsilon)$  requires degree at least $2n$  to refute in SOS \cite{Grigoriev2001a}. 
	
	Now, by Corollary \ref{degree-criterion}, to refute $\KNAPSACK(2n,n + \varepsilon)$  in SOS one needs at least $s_n$  many distinct significant monomials, where
	\[s_n =  \exp\left(\frac{(2n - 5)^2}{32(2n + 1)}\right).\]
	A lower bound of the same order was also obtained earlier in \cite{GrigorievHP02} using more ad hoc methods.
	
	%The lower bound on the number of monomials given by the size-degree trade-off is actually for a subset of all the explicit monomials in a refutation, namely for the number of monomials appearing in the polynomials $r_i$ and $t_q$ of a proof \eqref{proof}\footnote{This set of monomials was also called the set of explicit monomials of an SOS proof in \cite{AtseriasHakoniemi2019}.}. Moreover a minor modification of the proof of the trade-off gives a lower bound on the size of the set of all \emph{distinct} monomials in polynomials $r_i$ and $t_q$. This is because in the proof of Lemma 10 of \cite{AtseriasHakoniemi2019}, one can simply count the number of distinct monomials rather than the number of monomials with multiplicity, and the same proof goes through. 
	
	It follows that for any $\varepsilon$  strictly between $0$  and $1$  and for any set $S$  of monomials containing all the variables and the empty monomial $1$  of size less than $s_n$, there is no SOS refutation of $\KNAPSACK(2n, n + \varepsilon)$  over $S$, and thus, by Theorem \ref{soundnesscompleteness} there is an $S$-pseudoexpectation for $\KNAPSACK(2n, n + \varepsilon)$. This property will be key in the main result of this chapter.
	
	%By Theorem \ref{soundnesscompleteness}, for any set $S$  with less than $s_n$  monomials, containing all the variables and the empty monomial, and any $1 > \varepsilon > 0$, there exists a linear functional $E^\varepsilon_S\colon\R[S^2]\rightarrow\R$  such that
	%
	%\begin{romanenumerate}
	%	\item $E^\varepsilon_S(1) = 1$;
	%	\item $E^\varepsilon_S(p)\geq 0$  if $\KNAPSACK(2n,n + \varepsilon)\vdash_S p\geq 0$.
	%\end{romanenumerate}
	%\subsection{The Constraints}
	
	Now we turn to the set of constraints we consider here. The set is slightly modified from \cite{RaghavendraWeitz2017} --  we add the constraint \eqref{constraint-one} below in order to obtain an unsatisfiable set of constraints.
	
	For each $i\in [n]$, introduce $2n$  variables $x_{ij}$, $j\in [2n]$. Denote by $\ks_i$ the polynomial $\sum_{j\in [2n]} x_{ij} - n$. Note that the constraint $\ks_i = \varepsilon$ is just the constraint $\KNAPSACK(2n,n + \varepsilon)$  in the variables $x_{ij},j\in [2n]$. Denote by $Q_n$ the following set of constraints
	
	\begin{enumerate}[label=\Roman*),ref=\Roman*]
		\item $\ks_1 = 1/2$; \label{constraint-one}
		\item $\ks_i^2 = \ks_{i+1}$  for each $i\in [n-1]$; \label{constraint-two} 
		\item $\ks_n^2 = 0$. \label{constraints-three}
	\end{enumerate}
	
	Now, as noted above already, the constraint \eqref{constraint-one} is by itself unsatisfiable over the Boolean values. However, both PCR/$\Q$  and SOS require linear degree to refute \eqref{constraint-one} by itself. The role of \eqref{constraint-two} and \eqref{constraints-three} is two-fold: on one hand they decrease the degree needed to refute the constraints, but the repeated squaring inherent in the constraints also forces the coefficients to blow-up.
	
	%To prove our claims, we use a set of constraints slightly modified from \cite{RaghavendraWeitz2017}. For each $i\in [n]$, we introduce $2n$  variables $x_{ij}$, $j\in [2n]$. For any $i\in [n]$  denote by $\ks_i$  the polynomial $\sum_{j\in [2n]} x_{ij} - n$. That is, the constraint $\ks_i = \delta$  is equal to the constraint $\KNAPSACK(2n,n + \delta)$  in the variable $x_{ij},j\in [2n]$. Denote by
	%$Q_{n}$  the set of constraints consisting of
	%
	%\begin{romanenumerate}
	%	\item $\ks_i^2 = \ks_{i+1}$  for each $i\in [n-1]$;
	%	\item $\ks_n^2 = 0$;
	%	\item $\ks_1 = 1/2$.
	%\end{romanenumerate}
	%
	%Now the constraint (iii) is unsatisfiable over the Boolean cube by itself, but the other two constraints are needed to argue for the refutations of degree $2$  and of monomial-size polynomial in $n$.
	
	\section{Upper bounds}
	
	In this section we show that both PCR/$\Q$  and SOS have refutations of $Q_n$ of degree $2$  and of monomial-size polynomial in $n$. Each of these refutations however uses coefficients of exponential bit-complexity in $n$  and thus the refutations themselves have bit-complexities that are exponential in $n$.
	
	\begin{lemma}
		There is a PCR/$\Q$  refutation of $Q_n$ of degree $2$  and of monomial-size $\poly(n)$.
	\end{lemma}
	
	\begin{proof}
		We prove by induction that for any $i\in [n]$  there is a PCR/$\Q$  proof of $\ks_i^2 = 1/2^{2^{i}}$  from $Q_n$  of degree $2$  and monomial-size $\poly(n)$. For $i = 1$, we obtain this as follows. First derive $\ks_1^2 = \ks_1/2$  from $\ks_1 = 1/2$. This can be done in $O(n)$ steps, in degree $2$  and with polynomially many monomials. Secondly derive $\ks_1/2 = 1/4$  from $\ks_1 = 1/2$. This can be done in one step. Finally combine the two derivations to obtain a derivation of $\ks_1^2 = 1/4$.
		
		Suppose then that we have a proof of $\ks_i^2 = 1/2^{2^i}$  from $Q_n$. We derive $\ks_{i+1}^2 = 1/2^{2^{i+1}}$  as follows. First derive $\ks_{i+1} = 1/2^{2^i}$  from $\ks_i^2 = 1/2^{2^i}$  and $\ks_i^2 = \ks_{i+1}$. Secondly derive both $\ks_{i+1}^2 = \ks_{i+1}/2^{2^i}$  and $\ks_{i+1}/2^{2^i} = 1/2^{2^{i+1}}$  from $\ks_{i+1} = 1/2^{2^i}$, and combine these to obtain a proof of $\ks_{i+1}^2 = 1/2^{2^{i+1}}$.
		
		In the end we have a proof of $\ks_n^2 = 1/2^{2^n}$  from $Q_n$  of degree $2$  and of monomial-size $\poly(n)$. By combining this with the axiom $\ks_n^2 = 0$  we reach a contradiction. Note that this last step involves a multiplication by a coefficient of doubly exponential magnitude.
	\end{proof}
	
	\begin{lemma}
		There is an SOS refutation of $Q_n$ of degree $2$  and of monomial-size $\poly(n)$.
	\end{lemma}
	
	\begin{proof}
		The following is an SOS refutation of $Q_n$: 
		\begin{multline*}
		-1 = \sum_{i\in [n]}\frac{(1 - 2n^{2^{i-1}}\ks_i)^2}{n} \\- \sum_{i\in [n-1]}\left(4n^{2^i - 1}\left(\ks_i^2 - \ks_{i+1}\right)\right) \\- 4n^{2^n}\ks_n	+ 4\left(\ks_1 - \frac{1}{2}\right)
		\end{multline*}
		
	\end{proof}

	\section{Lower bound for Sums-of-Squares}
	
	In this section we prove our main claim about bit-complexity of SOS refutations. The proof of the claim is very similar to the one in \cite{RaghavendraWeitz2017} with the use of $S$-pseudoexpectations instead of degree bounded pseudoexpectations being the central novel idea in the proof.
	
	For a monomial $m$  in variables $x_{ij}, i\in [n],j\in [2n]$, for each $i\in [n]$  denote by $m_i$  the monomial in variables $x_{ij}, j\in [2n]$  such that $m = m_1\cdots m_n$. For any $I\subseteq [n]$  let $m_I = \prod_{i\in I}m_i$. We call $m_i$  and $m_I$  the projections of $m$  to index $i$  and set $I$, respectively. For a set $S$  of monomials in variables $x_{ij}, i\in [n],j\in [2n]$, denote by $S_i$  and $S_I$  the sets of projections of all elements of $S$  to index $i$  and set $I$, respectively. 
	
	\begin{theorem}\label{main-theorem-sos}
		There is a constant $c > 0$  such that for large enough $n$, any refutation of $Q_{n}$  has at least $2^{cn}$  distinct explicit monomials or contains a coefficient of magnitude at least $2^{2^n}/2^{cn}$.
	\end{theorem}
	
	\begin{proof}
		Let $c$  be such that $s_n \geq 2^{cn}$  for large enough $n$. Let $n$  be large enough, let $\Pi$  be an SOS refutation of $Q_{n}$  with less than $2^{cn}$  distinct explicit monomials and let $S$  be the set of explicit monomials appearing in the refutation $\Pi$. Now $S_i$  has size less than $s_n$  for any $i\in [n]$, and so, by Section \ref{section:knapsack} and by Theorem \ref{soundnesscompleteness}, for any $i\in [n]$  there is an $S_i$-pseudoexpectation $E_i$  for $\{\ks_i = 1/2^{2^{i-1}}\}$.
		
		Now define a linear functional $E\colon\R[S^2]\rightarrow\R$  as follows: for each $m\in S^2$  let
		\[
		E(m) : = E_1(m_1)\cdots E_n(m_n),
		\]
		and extend linearly to the whole of $\R[S^2]$.
		
		We prove that $E$  has the following properties:
		
		\begin{enumerate}[label=\roman*),ref=\roman*]
			\item $E(1) = 1$; \label{propertiesbititemone}
			%		\item $E(1) = 1$; 
			\item $E(m(x_{ij}^2 - x_{ij})) = 0$  for any $i\in [n]$, $j\in [2n]$  and $m\in S$; \label{propertiesbititemtwo}
			\item $E(m(x_{ij} + \bar{x}_{ij} - 1)) = 0$  for any $i\in [n]$, $j\in [2n]$  and $m\in S$; \label{propertiesbititemthree}
			\item $E(m(\ks_1 - 1/2)) = 0$ for any $m\in S$; \label{propertiesbititemsix}
			%			\item[(iii)] $E(m^2 - m) = 0$  for any $m\in S$;
			%		\item $E(m(\sum_{j\in [2n]} x_{1j} - (n + \delta))) = 0$  for any $m\in S$.
			\item $E(m(\ks_i^2 - \ks_{i + 1}))) = 0$  for any $i\in [n-1]$  and $m\in S$;\label{propertiesbititemfour}
			\item $|E(p\ks_n^2)|\leq |S|\Vert p\Vert/2^{2^{n}}$  for any polynomial $p\in\R[S]$; \label{propertiesbititemfive}
			\item $E(p^2)\geq 0$  for any $p\in \R[S]$. \label{propertiesbititemseven}
		\end{enumerate}\
		
		Now applying $E$  to the given refutation $\Pi$, we obtain that $-1 \geq E(p\ks_n^2)$, and thus 
		\[1 \leq |E(p\ks_n^2)| \leq |S|\Vert p\Vert/2^{2^{n}}.\] 
		By rearranging the inequality we obtain that 
		\[\Vert p\Vert \geq 2^{2^{n}}/|S|\geq 2^{2^{n}}/2^{cn}.\]
		
		Finally we prove that $E$  has the desired properties. \eqref{propertiesbititemone} follows since $E_i(1) = 1$  for any $i\in [n]$.
		
		For \eqref{propertiesbititemtwo}, we have that 
		\[E(m(x_{ij}^2 - x_{ij})) = E_i(m_i(x_{ij}^2 - x_{ij}))\prod_{i'\neq i}E_{i'}(m_{i'}) = 0,\] since $E_i(m_i(x_{ij}^2 - x_{ij})) = 0$  as $E_i$  is an $S_i$-pseudoexpectation for $\{\ks_i = 1/2^{2^{i - 1}}\}$  and $S_i$ contains both $m_ix_{ij}^2$  and $m_ix_{ij}$  by construction. The item \eqref{propertiesbititemthree} is proved similarly.
		
		For \eqref{propertiesbititemsix}, we have that 
		\[
		E(m(\ks_1 - 1/2)) = E_1(m_1(\ks_1 - 1/2))\prod_{i'\neq 1} E_{i'}(m_{i'}) = 0,
		\]
		since $E_1(m_1(\ks_1 - 1/2)) = 0$  as $E_1$  is an $S_1$-pseudoexpectation for $\{\ks_1 = 1/2\}$.
		
		%		For (iii), we have that $E(m^2) = \prod_{i} E_i(m_i^2) = \prod_i E_i(m_i) = E(m)$.
		%	For (iv), we have that
		%	\[
		%	E(m(\sum_{j\in [2n]} x_{1j} - (n + \delta))) = E_1(m_1(\sum_{j\in [2n]} x_{1j} - (n + \delta)))\prod_{i\neq 1}E_i(m_i) = 0,
		%	\] 
		%	where the second inequality follows from the fact that $E_1(m_1(\sum_{j\in [2n]} x_{1j} - (n + \delta))) = 0$  as $E_1$  is an $S_1$-pseudoexpectation for $\KNAPSACK(2n,n + \delta)$.
		%	
		For \eqref{propertiesbititemfour} and \eqref{propertiesbititemfive}, note that for any $i\in [n]$,
		\begin{align*}
		E(m\ks_i) & = E_i(m_i\ks_i)E(\prod_{i'\neq i}m_{i'})\\
		& = E_i(m_i/2^{2^{i-1}})E(\prod_{i'\neq i}m_{i'})\\
		& = E(m)/2^{2^{i-1}},
		\end{align*}
		where the second equality follows since $E_i(m_i(\ks_i -  1/2^{2^{i-1}})) = 0$  for any $m_i\in S_i$  as $E_i$  is an $S_i$-pseudoexpectation for $\{\ks_i = 1/2^{2^{i-1}}\}$. 
		
		Now for \eqref{propertiesbititemfour} we have that
		\begin{align*}
		E(m(\ks_i^2 - \ks_{i+1})) & = E(m\ks_i^2) - E(m\ks_{i+1})\\
		& = E(m)/(2^{2^{i-1}})^2 - E(m)/2^{2^i} = 0
		\end{align*}
		
		For \eqref{propertiesbititemfive}, write $p = \sum_{m\in S} a_m m$. Then
		\begin{align*}
		|E(p\ks_n^2)| & = |E(p)/(2^{2^{n-1}})^2|\\
		& \leq \sum_{m\in S} |a_m E(m)|/2^{2^n}\\
		& \leq|S|\Vert p\Vert/2^{2^n},
		\end{align*}
		where the last inequality follows from the fact that $0\leq E(m)\leq 1$  for any $m\in S$  because $m \equiv m^2 \mod I_n$  and $1 - m \equiv (1 - m)^2 \mod I_n$.
		
		Finally to see that \eqref{propertiesbititemseven} holds, define for each $i\in [n]$, a linear function $T_i$  with $T_i(m) = E_i(m_i)\prod_{i'\neq i}m_{i'}$. Now clearly $E(m) = T_1(T_2(\ldots T_n(m)\ldots))$. We show that for any $i\in [n]$  and any $p\in\R[S_{[i]}]$, $T_i(p^2)$  is a sum of squares of polynomials in $\R[S_{[i-1]}]$, where $S_{[i]}$  is the projection of $S$  to the initial segment $[i]$. To simplify notation, we prove the case when $i = 2$. The general case is not conceptually any harder. So write $p$  as
		\[
		\sum_{\alpha}\sum_\beta a_{\alpha\beta} x_1^\alpha x_2^\beta,	
		\]
		where $x_1$  and $x_2$  are sequences of the variables in $S_1$  and $S_2$, respectively. Now
		\[
		T_2(p^2) = \sum_{\alpha,\alpha'}\sum_{\beta,\beta'} a_{\alpha\beta}a_{\alpha'\beta'} x_1^\alpha x_1^{\alpha'} E_2(x_2^\beta x_2^{\beta'}).
		\]
		Now the matrix $(E_2(x_2^\beta x_2^{\beta'}))_{\beta,\beta'}$  is positive semidefinite, and so there are some vectors $u$  such that $E_2(x_2^\beta x_2^{\beta'}) = \sum_u u_\beta u_{\beta'}$. Now
		\begin{align*}
		T_2(p^2) & = \sum_{\alpha,\alpha'}\sum_{\beta,\beta'} a_{\alpha\beta}a_{\alpha'\beta'} x_1^\alpha x_1^{\alpha'} \sum_u u_\beta u_{\beta'}\\
		& = \sum_{\alpha,\alpha'}(\sum_{\beta} \sum_u a_{\alpha\beta}u_\beta)(\sum_{\beta'}\sum_u a_{\alpha'\beta'}u_{\beta'}) x_1^\alpha x_1^{\alpha'}\\
		& = (\sum_{\alpha}\sum_\beta\sum_u a_{\alpha\beta}u_\beta x^\alpha_1)^2
		\end{align*}
	\end{proof}
	
	As a corollary to the above theorem, we obtain the following lower bound for the bit-complexity of SOS refutations.
	
	\begin{corollary}\label{sos-corollary}
		Any SOS refutation of $Q_n$  has bit-complexity $2^{\Omega(n)}$.
	\end{corollary}
	
	%Secondly as a corollary to the above corollary and to the p-simulation of PCR/$\Q$  by SOS \cite{Berkholz} we obtain the following lower bound for the bit-complexity of PCR/$\Q$  refutations. The lower bound below is of different character than the one for SOS. For SOS we can pin-point exactly where the large coefficients must reside when there are too few monomials in use. However for PCR/$\Q$  this is not the case, and we obtain the lower bound only as a consequence of the p-simulation.
	%
	%\begin{corollary}
	%	Any PCR/$\Q$ refutation of $Q_n$ has bit-complexity $2^{\Omega(n)}$.
	%\end{corollary}
	
	\section{Lower bound for Polynomial Calculus}
	
	Finally in this section we prove an analogue of Theorem \ref{main-theorem-sos} for Polynomial Calculus Resolution over rationals. Already from the Corollary \ref{sos-corollary} alone we obtain lower bounds on the bit-complexity of PCR/$\Q$  refutations of $Q_n$  using the simulation of \cite{Berkholz}. It is however instructive to prove an analogue of Theorem \ref{main-theorem-sos} also for PCR/$\Q$. 
	
	For SOS we were able to pinpoint exactly where the large coefficient resides in an SOS refutation that uses too few monomials: it must reside in the lift of the constraint $\ks_n^2 = 0$. However for PCR/$\Q$  we will not be able to be this precise. Moreover we need to bring height of the refutation also into the picture. We show that any PCR/$\Q$  refutation that uses only few monomials and coefficients of small magnitude must be very tall.
	
	To prove the theorem for Polynomial Calculus we first prove a form of simulation between SOS and PCR/$\Q$  that gives explicit bounds on coefficients in the SOS simulation in terms of the height of a given PCR/$\Q$ refutation. For the lemma we say that a PCR/$\Q$  proof is $R$-bounded for $R > 0$, if every coefficient in every polynomial of the proof is bounded from above by $R$   in absolute value, and the scalars $a$  and $b$  used in each instance of linear combination rule are also bounded from above by $R$  in absolute value.
	
	In order to state the following lemma we need to define a set $\tilde{S}$  of monomials from a given set $S$. The motivation for this definition here is purely technical. We need $\tilde{S}$  to contain the whole of $S^2$  and enough other monomials so that for every $m\in S$  there are $u_i,v_i\in\R[\tilde{S}]$  so that 
	\[
	m - m^2 = \sum_{i\in [n]}\left(u_i(x_i^2 - x_i) + v_i(x_i + \bar{x}_i - 1)\right).
	\]
	It is clear that there is $\tilde{S}$  satisfying the two requirements that is of size polynomial in the size of $S$ and in the maximum degree of a monomial in $S$. We  can assume without a loss of generality that the maximum degree of a monomial in $S$ is at most linear in $n$.
	
	\begin{lemma}\label{simulation}
		Let $Q$  be a set of polynomial equality constraints, and let $R \geq 2$. Suppose there is an $R$-bounded PCR/$\Q$  refutation of $Q$  of height $h$  that uses only monomials from a set $S$. Then there are polynomials $r_i\in\R[S]$  and scalars $a_q\in\R$  for every $q\in Q$  and polynomials $u_i,v_i\in\R[\tilde{S}]$  for $i\in [n]$  such that
		\[
		-1 = \sum r_i^2 + \sum_{q\in Q}a_q q^2 + \sum_{i\in [n]}\left(u_i(x_i^2 - x_i) + v_i(x_i + \bar{x}_i - 1)\right) 
		\] 
		with $|a_q|\leq R^{4(h + 1)}$  for every $q\in Q$.
	\end{lemma}
	
	\begin{proof}
		Suppose towards a contradiction that the above claim does not hold. Then the following sets are disjoint: 
		\begin{multline*}
		A := \{p\in\R[S^2] : p = \sum r_i^2 + \sum_{i\in [n]}\left(u_i(x_i^2 - x_i) + v_i(x_i + \bar{x}_i - 1)\right),\\\text{where } r_i\in\R[S]\text{ and }u_i,v_i\in\R[\tilde{S}]\text{ for every }i\in [n]\}
		\end{multline*} 
		and 
		\begin{equation*}
		B := \{-1 + \sum_{q\in Q} a_q q^2 : |a_q|\leq R^{4(h + 1)}\}.
		\end{equation*}
		
		Now $A$  is a convex cone and $B$  is a convex set. By the hyperplane separation theorem there is a non-trivial linear functional $E\colon\R[S^2]\rightarrow\R$  such that $E(p)\geq 0$  for every $p\in A$  and $E(p')\leq 0$  for every $p'\in B$. 
		
		We argue first that $E(1) > 0$. Suppose towards a contradiction that $E(1) = 0$. We show that then $E(m) = 0$  for any $m\in S^2$  against the non-triviality of $L$.
		
		Let first $m\in S$. Now, by construction there are $u_i,v_i\in\R[\tilde{S}]$  so that 
		\[m - m^2 = \sum_{i\in [n]}\left(u_i(x_i^2 - x_i) + v_i(x_i + \bar{x}_i - 1)\right),\]
		and thus both $m - m^2$  and $m^2 - m$  are in $A$. Hence $E(m) = E(m^2)$  for any $m\in S$.
		
		Now, since $1 - m = (1 - m)^2 + (m - m^2)$  we have that $E(1)\geq E(m)$, and thus $0\geq E(m)$. On the other hand as $E(m^2)\geq 0$ also $E(m)\geq 0$, and thus $E(m) = 0$ for every $m\in S$.
		
		Let then $m_1,m_2\in S$. Now $m_1^2 \pm 2m_1m_2 + m_2^2 = (m_1 \pm m_2)^2$ and so $|E(2m_1m_2)|\leq E(m_1^2) + E(m_2^2) = 0$. Hence $E(1) > 0$, and we can assume by scaling that $E(1) = 1$.
		
		Let then $p_1,\ldots,p_\ell$  be an $R$-bounded PCR/$\Q$  refutation of $Q$ of height $h$  using only monomials from a set $S$. We prove by induction on the structure of the refutation that for any $p_i$  at height $h'$  we have that $E(p_i^2) \leq 1/R^{4(h - h' + 1)}$. The claim holds clearly for any Boolean axiom.
		
		The claim holds for any $q\in Q$, since any $q\in Q$  is at height $0$  and we have that $-1 + R^{4(h + 1)}q^2\in B$, and thus $E(-1 + R^{4(h+1)}q^2)\leq 0$, i.e. $E(q^2)\leq 1/R^{4(h+1)}$.
		
		Suppose that $p_i$  at height $h' + 1$  is obtained from $p_j$  via a lift with a variable $x$, i.e. $p_i = xp_j$  for some $x$. Now $E((xp_j)^2)\leq E(p_j^2)$, since $p_j^2 - (xp_j)^2 = (p_j - xp_j)^2 - 2p_j^2(x^2 - x)$. Now $p_j$  is at height $h'$, and so by induction assumption, $E(p_j^2) \leq 1/R^{4(h - h' + 1)}$. Hence $E((xp_j)^2)\leq 1/R^{4(h - h')}$.
		
		Suppose then that $p_i$  at height $h' + 1$  is obtained from $p_j$ and $p_k$ via linear combination, i.e. that there are some $a,b\in\Q$  such that $p_i = ap_j + bp_k$. Now both $p_j$ and $p_k$  are at most at height $h'$, and so, by induction assumption, $E(p_j^2),E(p_k^2)\leq 1/R^{4(h - h' + 1)}$. Secondly, by assumption, $|a|,|b|\leq R$, and so $a^2,b^2\leq R^2$. Hence $E(a^2p_j^2),E(b^2p_k^2)\leq R^2/R^{4(h - h' + 1)}$. Thirdly 
		\[a^2p_j^2 - 2abp_jp_k + b^2p_k^2 = (ap_j - bp_k)^2,\] 
		and so $
		E(2abp_jp_k)\leq E(a^2p_j^2) + E(b^2p_k^2)$. Now
		\begin{align*}
		E(p_i^2) & = E(a^2p_j^2) + E( 2abp_jp_k) + E(b^2p_k^2)\\ 
		& \leq 2\left(E(a^2p_j^2) + E(b^2p_k^2)\right)\\
		& \leq 4R^2/R^{4(h - h' + 1)}\\
		& \leq R^4/R^{4(h -h' + 1)}\\
		& = 1/R^{4(h - h' )},
		\end{align*}
		where the fourth line follows, since $R^2\geq 4$.
		
		%	
		%	Suppose then that $p_i$ at height $h' + 1$ is obtained from $p_j$  via multiplication by a scalar $a$, i.e. $p_i = ap_j$  for some $a$. Now $|a|\leq R$ and so $a^2 \leq R^2$. Now by induction assumption $E(p_j^2)\leq 1/R^{2(h - h' + 1)}$, and so $E((ap_j)^2)\leq R^2/R^{2(h - h' + 1)} \leq 1/R^{2(h - h')}$.
		%	
		%	Suppose finally that $p_i$ at height $h' + 1$ is obtained from $p_j$  and $p_k$ by addition, i.e. $p_i = p_j + p_k$  for some $j,k < i$. Now $p_j^2 - 2p_jp_k + p_k^2 \equiv (p_j - p_k)^2 \mod I_n$ and so $E(2p_jp_k)\leq E(p_j^2) + E(p_k^2)$. Now by induction assumption $E(p_j^2)\leq 1/R^{2(h - h' + 1)}$ and $E(p_k^2)\leq 1/R^{2(h - h' + 1)}$. Now
		%	\[
		%	E((p_j + p_k)^2) = E(p_j^2) + E(2p_jp_k) + E(p_k^2)\leq 4/R^{2(h - h' + 1)}\leq 1/R^{2(h - h')}.
		%	\]
		Now $E(1)\leq 1/R$  against the assumption that $E(1) = 1$.
	\end{proof}

	\begin{theorem}
		There are constants $c > 0$  and $d > 0$  such that for large enough $n$, every $2^{2^{n/2}}$-bounded PCR/$\Q$  refutation of $Q_n$  that uses at most $2^{dn}$  different monomials has height at least
		\[
		2^{n/2 - 2} - \frac{cn + 2\log n}{2^{n/2 + 2}} - 1.
		\]
	\end{theorem}
	
	\begin{proof}
		Let $c$  be as in Theorem \ref{main-theorem-sos}, and let $d$  be such that for any set $S$  of monomials of size at most $2^{dn}$, the size of $\tilde{S}$  is less than $2^{cn}$  for large enough $n$. Let $n$  be large enough and let $\Pi$  be a $2^{2^{n/2}}$-bounded PCR/$\Q$  refutation of $Q_n$  of height $h$  that uses at most $2^{dn}$  different monomials. Let $S$  be the set of all monomials in the refutation. Now by Lemma \ref{simulation}, there are polynomials $r_i\in\R[S]$, scalars $a_q\in\R$  for every $q\in Q_n$, and polynomials $u_i,v_i\in\R[\tilde{S}]$  such that
		\begin{equation*}
		-1 = \sum r_i^2 + \sum_{q\in Q_n} a_q q^2 + \sum_{i\in [n]}\left(u_i(x_i^2 - x_i) + v_i(x_i + \bar{x}_i - 1)\right),
		\end{equation*}
		where $|a_q|\leq 2^{2^{n/2}4(h + 1)}$  for every $q\in Q_n$.
		
		The explicit monomials of the above SOS refutation are among $\tilde{S}$  and the size of $\tilde{S}$  is less that $2^{cn}$. Thus, by the proof of Theorem \ref{main-theorem-sos}, the lift of the constraint $\ks_n^2 = 0$  contains a coefficient of magnitude at least $2^{2^n}/2^{cn}$. On the other hand, since $|a_q|\leq 2^{2^{n/2}4(h + 1)}$ the coefficients in lift of the constraint $\ks_n^2 = 0$, i.e. in $a_q\ks_n^2$  are bounded from above by $2^{2^{n/2}4(h + 1)} n^2$  in absolute value.
		
		Putting everything together we obtain that
		\[
		2^{2^{n/2}4(h + 1)} n^2\geq 2^{2^n}/2^{cn}.
		\]
		After solving for $h$  we obtain the wanted lower bound for the height.
	\end{proof}
	
	We obtain a lower bound on bit-complexity for Polynomial Calculus as a corollary to the above theorem.
	
	\begin{corollary}
		Any PCR/$\Q$  refutation of $Q_n$  has bit-complexity $2^{\Omega(n)}$.
	\end{corollary}
	
	%This corollary should be contrasted with the polynomial time decision procedure for the \emph{existence} of a bounded degree PCR/$\Q$  refutations given in \cite{DBLP:journals/lmcs/GradelGPP19}. We have an interesting situation here with PCR/$\Q$: although the proof search itself can take exponential amount of time, deciding whether there exists a bounded degree PCR/$\Q$  refutation can always be done in polynomial time. Both algorithms, the one in \cite{CleggEI96} and the one in \cite{DBLP:journals/lmcs/GradelGPP19}, work by constructing a a generating set for the subspace of all polynomials provable in PCR/$\Q$  in some fixed degree $d$. While the algorithm in \cite{CleggEI96}, in essence, builds PCR/$\Q$ proofs for all the elements of the generating set, and thus provides us with a \emph{proof search} algorithm for PCR/$\Q$, the algorithm of \cite{DBLP:journals/lmcs/GradelGPP19} involves a step that does not respect the proof system. Namely, the algorithm involves a step of computing auxiliary generating sets from the one under construction, and using those to expand the main set. This computation does not respect the proof system in the sense that there is no guarantee that the elements of the auxiliary sets have small or low-degree proofs from the elements of the main one.
	
	\section{Conclusions and open questions}
	
	We have shown here that there is an example of a set of constraints that has both SOS and PCR/$\Q$  refutations with polynomially many monomials, but for which any refutation must have exponential bit-complexity.
	
	The most important open question related to the ideas in this paper is whether the phenomena observed here can occur when the set of constraints comes from a translation of a CNF, or whether the two measures of bit-complexity and monomial-size are polynomially equivalent, when SOS or PCR/$\Q$  are considered as a refutation system for CNFs. The constraints in this paper do not arise from any CNF. 
	
	\section*{Acknowledgment}
	
	I want to thank Albert Atserias for discussions on preliminary versions of this paper. I would also like to thank the anonymous reviewers for many helpful comments that have improved the presentation above.
	
	
	\bibliographystyle{plain}
	\bibliography{bibliography}
\end{document}